\documentclass[
  aps,
  prd,
  twocolumn,
  nofootinbib,
  superscriptaddress
]{revtex4-2}

\usepackage{amsmath,amssymb,bm}
\usepackage{ulem}
\usepackage{dcolumn}
\usepackage{mathrsfs}
\usepackage{hyperref}
\usepackage{graphicx}
\usepackage{float}
\usepackage{tikz}
\usetikzlibrary{arrows.meta,positioning,calc}
\usepackage{microtype}
\usepackage{xcolor}
\usepackage{cancel}
\usepackage{changes}

\newcommand{\dd}{\mathrm d}
\newcommand{\Lie}{\mathcal L}
\newcommand{\eps}{\epsilon}

\newtheorem{theorem}{Theorem}
\newtheorem{definition}{Definition}

\begin{document}

\title{Constraints on the number of inner horizons from gravitational Noether charges}

\author{Jing-Peng Ye}
%\email{yejingpeng@tju.edu.cn}
\affiliation{Center for Joint Quantum Studies and Department of Physics, School of Science, Tianjin University, Yaguan Road 135, Jinnan District, 300350 Tianjin, P. R. China}

\author{Run-Qiu Yang}
\email{aqiu@tju.edu.cn}
\affiliation{Center for Joint Quantum Studies and Department of Physics, School of Science, Tianjin University, Yaguan Road 135, Jinnan District, 300350 Tianjin, P. R. China}

\begin{abstract}
A central question in the study of black hole interiors is how matter and geometry constrain the number of horizons. In Einstein gravity, one known method is to use the ADM decomposition of the field equations to derive a differential inequality for the lapse, which then constrains the possible existence of multiple inner horizons. We show that the same obstruction can instead be formulated through the Komar integral as a bulk-boundary balance of gravitational Noether charge. In higher-curvature gravity, a direct ADM analysis becomes more complicated, whereas the gravitational Noether charge naturally retains this bulk-boundary structure. We naturally apply the gravitational Noether charge analysis to Lovelock gravity. We introduce a TCC-like curvature condition and prove that, together with the thermodynamically motivated assumption of positive Wald entropy, it allows at most one nondegenerate inner Killing horizon inside a nondegenerate event horizon. This suggests that the obstruction reflects a compatibility between bulk gravitational response and horizon thermodynamics.

\end{abstract}

\maketitle

\section{Introduction}

\label{sec:intro}

Black holes are fascinating extreme objects predicted by general relativity and among the most important subjects in theoretical physics. From LIGO's first detection of gravitational waves from a binary black hole merger\cite{LIGOScientific:2016aoc} to the images of M87*~\cite{EventHorizonTelescope:2019dse} and Galactic Center black hole released by the Event Horizon Telescope (EHT)~\cite{EventHorizonTelescope:2022wkp}, black holes have transformed from purely mathematical objects into observable real astrophysical objects. A central feature of a black hole is the existence of a causal boundary: the event horizon, which prevents interior signals from reaching future null infinity. This does not mean that the black hole interior is physically unimportant. On the contrary, the black-hole interior plays a central role in theoretical physics. Its study is closely related to a number of fundamental problems, including the global causal structure of classical black hole spacetimes~\cite{Carter:1968rr,Hawking:1974ar}, Cauchy horizons and the strong cosmic censorship (SCC) conjecture~\cite{Penrose:1969pc,Davey:2024xvd}, the Hawking–Penrose singularity theorems~\cite{Hawking:1970zqf}, blueshift instability and mass inflation near inner horizons~\cite{Poisson:1989zz,Poisson:1990eh,Ori:1991zz}, and Kasner/BKL dynamics near the singularity region~\cite{Belinsky:1970ew,Cai:2023igv,duan2026kasnerepochsorderedoscillations,DeClerck:2023fax,Arean:2024pzo,Caceres:2026mug}. Understanding the geometry and causal structure of the black hole interior remains an important problem.

As horizons are the basic geometric objects characterizing the causal structure of the black hole interior, a natural geometric and physical question is: how many different horizons can exist inside a black hole? It is well known that the Schwarzschild black hole has no inner horizon, whereas the Reissner–Nordstr\"{o}m and Kerr black holes each possess one inner horizon inside the event horizon. Some general matter models or gravitational theories can admit black hole solutions with multiple inner horizons ~\cite{Gao:2017vqv,Gao:2021kvr,Camanho:2011rj,Kastor:1992nn,Nakao:1994mm,Bronnikov:2012mf,Nojiri:2017kex}. For the horizon structure inside black holes, a series of studies on no-inner horizon theorems and constraints on the number of inner horizons has emerged in recent years. Many works found in specific matter models that nontrivial scalar hair, vector hair, or nonlinear electrodynamics couplings lead to strict no-inner horizon conclusions~\cite{Hartnoll:2020rwq,Cai:2020wrp,Dias:2021afz,Devecioglu:2023hmn}. Subsequently, some studies further moved beyond specific matter models and instead used energy conditions and the geometric structure of the black hole interior to directly constrain the number of inner horizons~\cite{yang2022horizons,ye2026innerkillinghorizonstationary}. In particular, it was proved that under the strong energy condition (SEC) or the null energy condition (NEC), at most one nondegenerate inner Killing horizon can exist inside a nondegenerate event horizon. These results show that the horizon structure inside a black hole is not arbitrary, but is jointly constrained by the properties of matter, energy conditions, and the structure of the gravitational field equations.

However, the above constraints on inner horizons are mainly established within Einstein gravity. From the viewpoint of effective field theory, general relativity can be regarded as the leading effective description of gravity at low energies and low curvature scales. When ultraviolet or quantum effects at higher energy scales are considered, higher-curvature corrections generally appear in the gravitational effective action~\cite{Donoghue:1994dn,Burgess:2003jk,Zwiebach:1985uq,Boulware:1985wk}. Therefore, when the gravitational action contains higher-curvature corrections, do the constraints on the number of horizons inside a black hole still hold? This question remains open. Related no-inner horizon theorems have been obtained for specific scalar-hair models with Gauss–Bonnet couplings \cite{Devecioglu:2023hmn}, using conserved charges associated with scaling symmetries of the reduced action. Unlike these model dependent results, our aim is to establish a geometric constraint on the number of inner Killing horizons without specifying the matter Lagrangian.

To answer this question, we investigate the constraints on inner Killing horizons inside static black holes in Lovelock gravity. In Einstein gravity, besides the local argument based on the lapse equation, the constraint on the number of inner Killing horizons admits an equivalent formulation in terms of an integrated Komar charge balance law. 
%\sout{besides the local argument based on the lapse equation, the Einstein inner horizon constraint admits an equivalent formulation in terms of an integrated Komar charge balance law.} 
%\sout{On this basis,} 
We show that, in Einstein gravity, the Komar charge coincides with the gravitational Noether charge associated with the diffeomorphism symmetry generated by a Killing vector. In higher-curvature gravity, the ADM decomposition generally leads to lapse equations involving complicated curvature contractions, making it difficult to establish analogous constraints on the number of inner horizons. By contrast, the Komar charge formulation admits a natural extension through the gravitational Noether charge, whose expression remains particularly simple in Lovelock gravity. %{\color{red}(Comment by YRQ: Add sentences to mention three aspects: (1) modified gravity is important! (2) ADM decomposition is complicated and how to restrict number of horizon based this approach is unknown (3) ``Komar charge'' is easy to be generalized in to modified gravity theories.)}% 
We therefore construct the corresponding  balance law based on gravitational Noether charge in Lovelock gravity.  On the one hand, this balance law naturally identifies the Ricci-like curvature tensor controlling the bulk contribution. On the other hand, it also relates the Noether charge on the horizon to the surface gravity and Wald entropy. We thereby prove that, under an appropriate TCC-like curvature condition and positive Wald entropy, a static Lovelock black hole with a nondegenerate event horizon can contain at most one nondegenerate inner Killing horizon. Our result suggests that the obstruction to multiple inner horizons is not a specific feature of the Einstein equations, but reflects a deeper compatibility condition between the gravitational response in the bulk and the thermodynamic data carried by the inner horizons. In this work, we use units in which $c=\hbar=G=1$.

\section{Set up and geometric assumptions}
\label{sec:pocket}

Consider a $D=d+1$ dimensional static spacetime $(\mathcal M,g_{\mu\nu})$ with a Killing vector field $\xi^{\mu}$ associated with the static symmetry. Choose the time coordinate $t$ such that $(\partial_t)^{\mu}=\xi^{\mu}$ and foliate the spacetime $\mathcal{M}$ by a series of constant-$t$ slices $\{\Sigma_{t}\}$. The metric then has the following $d+1$ decomposition:
\begin{equation}\label{metric}
    ds^2=-N^2\mathrm{d}t^{2}+h_{ab}\mathrm{d}x^{a}\mathrm{d}x^{b},
\end{equation}
where $N$ is the lapse function and $h_{ab}$ is the induced metric on $\Sigma_{t}$. We adopt the following index convention: Greek indices $\mu,\nu=0,1,\dots,d$ are spacetime indices and Latin indices $a,b=1,2,\dots,d$ are spatial indices on $\Sigma_{t}$.

In a static spacetime, Killing horizons are particularly natural to study because they possess a well-defined equilibrium thermodynamic description, characterized by quantities such as surface gravity and entropy. A Killing horizon $\mathcal{H}$ is a null hypersurface generated by the Killing vector $\xi^{\mu}$, whose norm is the negative of the square lapse (i.e., $\xi^{\mu}\xi_{\mu}=-N^2$). Hence, the horizon is located at the zeros of the lapse function: $N^{2}|_{\mathcal{H}}=0$. In this work, we assume that the Killing horizons under consideration are regular and admit compact spatial cross sections.

To specify the black hole interior relevant to our discussion, we follow Ref.~\cite{yang2022horizons}.
We choose a single connected component of the future event horizon and define the associated black hole interior as the region reached by future-directed ingoing null geodesics emanating from it. As illustrated in Fig.~\ref{fig:Black_hole_interior}, the Schwarzschild black hole has no inner Killing horizon, whereas the Kerr-Newman black hole contains at most one inner Killing horizon. Suppose that two consecutive nondegenerate inner Killing horizons $\mathcal{H}_{1}$ and $\mathcal{H}_{2}$, where $\mathcal{H}_{1}$ lies closer to the event horizon than $\mathcal{H}_{2}$, occur on the same black hole interior branch, as illustrated in Fig.~\ref{fig:pocket}. The portion of spacetime enclosed by these two horizons is denoted by $\mathcal V$. Since $\xi^\mu$ is timelike throughout $\mathcal V$, we call $\mathcal V$ a static region. On the constant-$t$ hypersurface $\Sigma_{t}$, the spatial section of the static region $\mathcal{D}:=\Sigma_{t}\cap \mathcal{V}$ obeys
\begin{equation}
N^2=0\quad\hbox{on}\quad\partial\mathcal D,\qquad
N^2>0\quad\hbox{in}\quad{\mathring{\mathcal D}}.
\label{eq:pocket}
\end{equation}
Here, $\partial\mathcal D=(\mathcal H_1\cup\mathcal H_2)\cap\Sigma_t$, and $\mathring{\mathcal D}$ denotes the interior of $\mathcal D$.

\begin{figure}[t]
    \centering
    \includegraphics[width=0.65\linewidth]{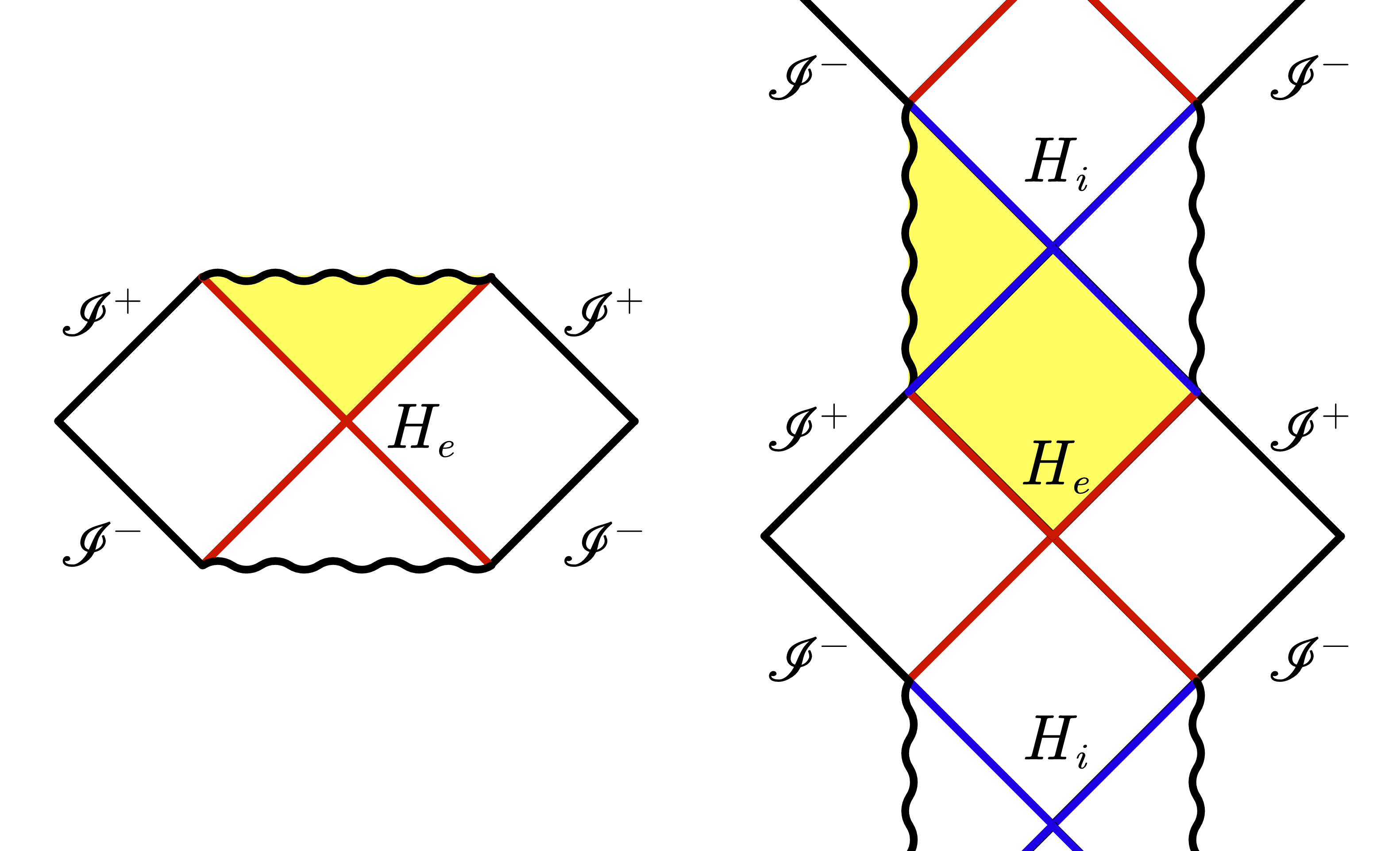}
    \caption{Left: Penrose diagrams of Schwarzschild black hole. Right: Penrose diagrams of Kerr-Newman black hole. The red and blue lines denote the event and inner Killing horizons, respectively, while the yellow region denotes the black hole interior.}
\label{fig:Black_hole_interior}
\end{figure}

\begin{figure}[t]
    \centering
    \includegraphics[width=0.65\linewidth]{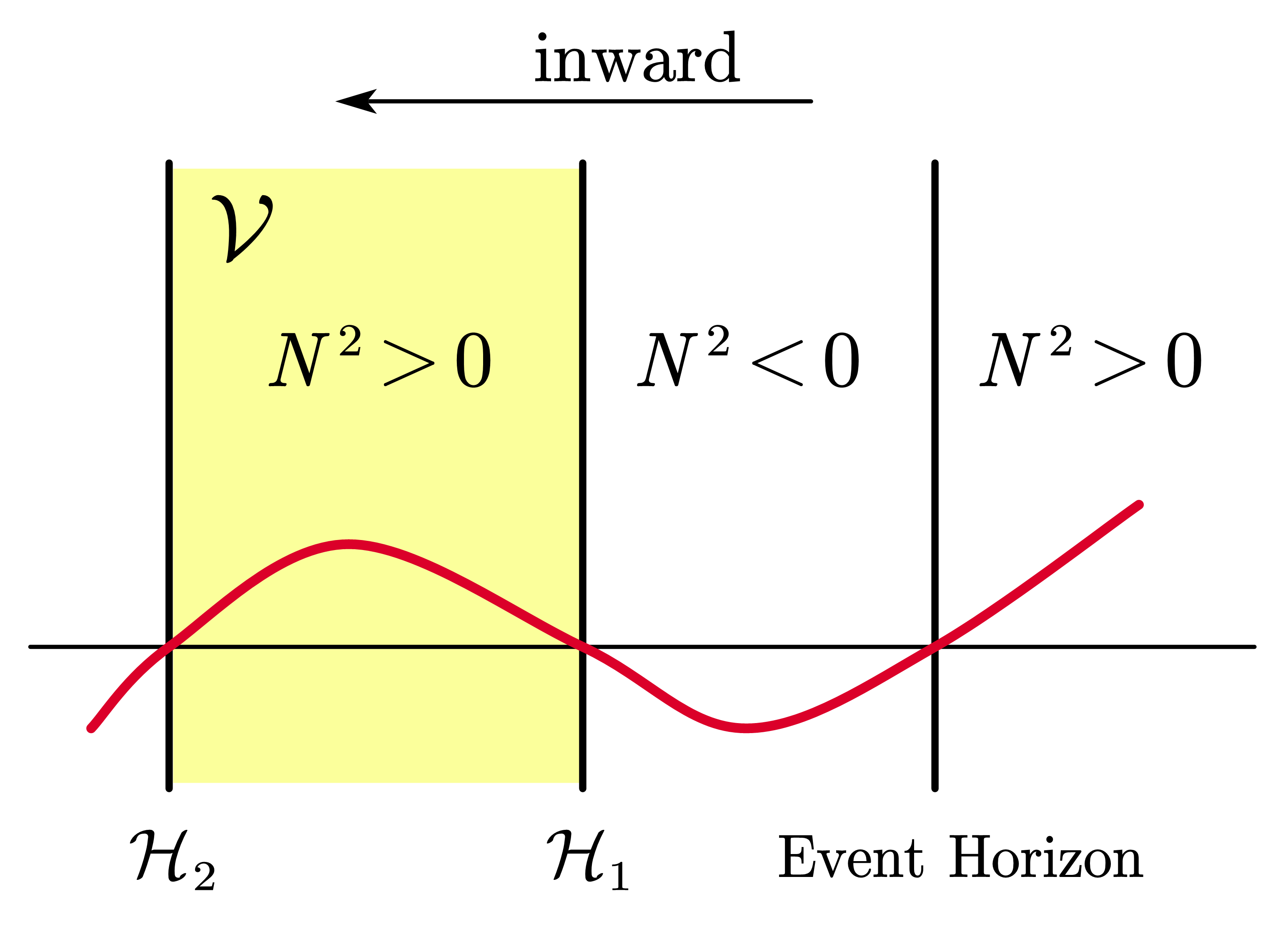}
    \caption{Schematic lapse structure inside a static black hole with two consecutive nondegenerate inner Killing horizons $\mathcal H_1$ and $\mathcal H_2$. The yellow region $\mathcal V$ between them is static region, with $N^2>0$.}
\label{fig:pocket}
\end{figure}

To characterize the two Killing horizons bounding the static region $\mathcal V$, we introduce their surface gravities, whose signs will play a central role in the argument below. The surface gravity $\kappa$ is defined as
\begin{equation}\label{SG}
    \xi^{\nu}\nabla_{\nu}\xi^{\mu}=\kappa \xi^{\mu}\qquad \text{on}\ \mathcal{H}.
\end{equation}
In the static case considered here, the zeroth law of black hole mechanics~\cite{Bardeen:1973gs} states that the surface gravity $\kappa$ is constant over each connected Killing horizon. If $\kappa=0$ on $\mathcal{H}$, this Killing horizon is said to be a degenerate horizon. In the present study, we are interested in black holes with nondegenerate Killing horizons. Under the regular cross-section and orbit assumption of R\'{a}cz and Wald, a regular Killing horizon with nonzero constant surface gravity admits a local extension in which it becomes a branch of a regular bifurcate Killing horizon~\cite{Racz:1992bp}. Moreover, R\'{a}cz and Wald showed that, in a static spacetime with a bifurcate Killing horizon, the hypersurfaces orthogonal to the Killing field $\xi^{\mu}$ extend smoothly to the bifurcation surface. For the two boundary horizons $\mathcal H_1$ and $\mathcal H_2$, let
$B_1$ and $B_2$ denote the corresponding bifurcation surfaces, as
illustrated in Fig.~\ref{fig:bifurcation}. We assume that the hypersurface $\mathcal D$ extends smoothly to the bifurcation surfaces $B_1$ and $B_2$ associated with the two boundary horizons, so that
\begin{equation}
    \partial \mathcal D=B_1\cup B_2 .
\end{equation}
On bifurcation $B$, the Killing vector vanishes (i.e., $\xi^{\mu}\big|_{B}=0$). To fix orientation conventions, let $n^{\mu}$ be a future-directed unit normal to $\Sigma_{t}$ and let $s^a$ be the unit spacelike normal to the horizon section within $\Sigma_t$, chosen to point toward the event horizon. The binormal is then
\begin{equation}
 \eps_{\mu\nu}=2n_{[\mu}s_{\nu]},
 \qquad
 \eps_{\mu\nu}\eps^{\mu\nu}=-2.
 \label{eq:binormal}
\end{equation}
Since $\xi^{\mu}$ vanishes on the bifurcation surface $B$, its derivative there acts purely in two-dimensional normal plane on $B$. With the binormal defined, the same surface gravity is therefore characterized by
\begin{equation}\label{eq:bifurcation-kappa}
    \nabla_\mu\xi_\nu=\kappa\eps_{\mu\nu}.
\end{equation}

\begin{figure}[t]
    \centering
    \includegraphics[width=0.6\linewidth]{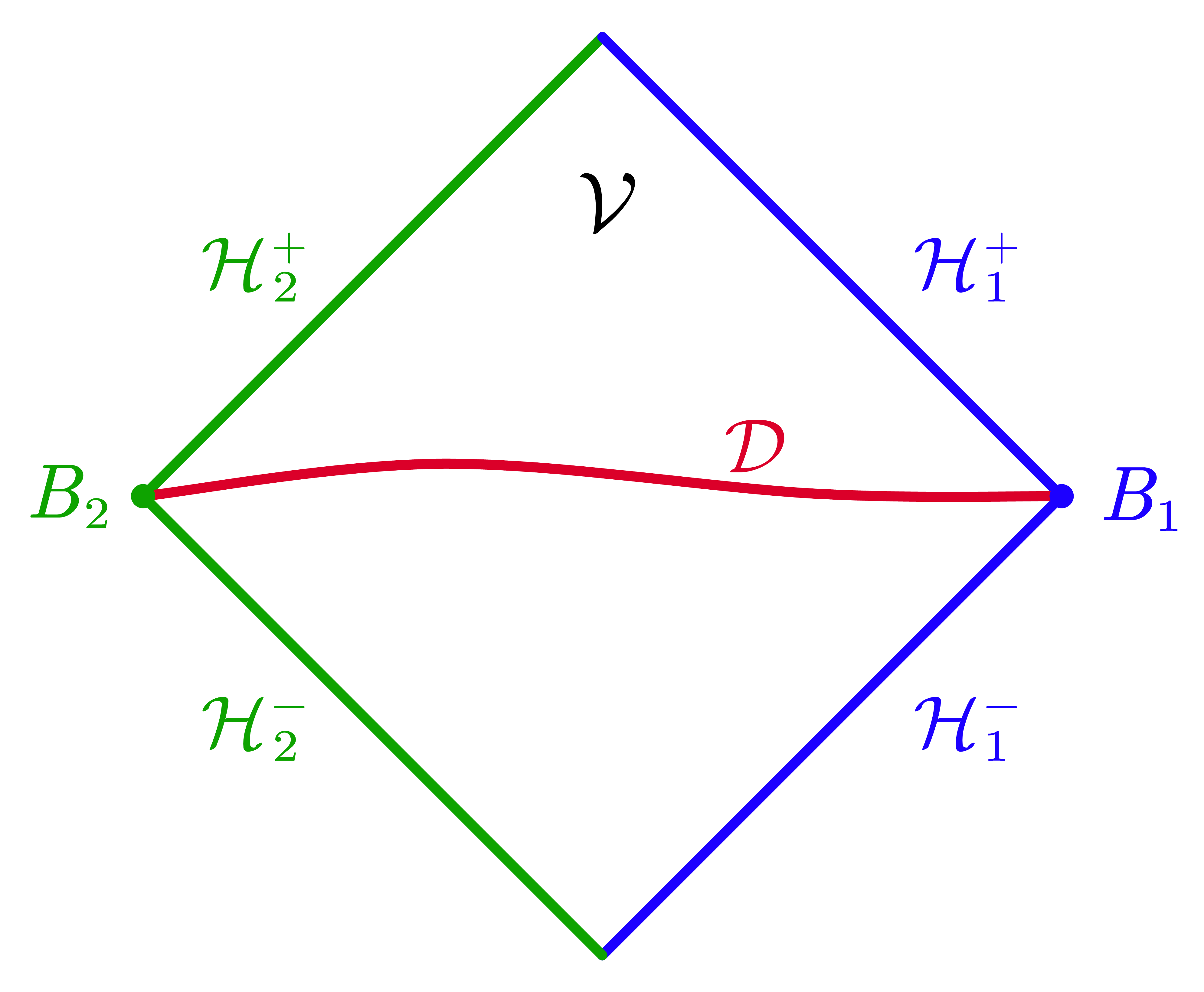}
    \caption{Schematic bifurcate extension of the static region $\mathcal V$ bounded by the two nondegenerate inner Killing horizons $\mathcal H_1$ and $\mathcal H_2$. The green branches denote $\mathcal H_2^\pm$, which intersect at the bifurcation surface $B_2$, while the blue branches denote $\mathcal H_1^\pm$, which intersect at $B_1$. The red curve represents the completed static slice $\mathcal D$, which connects the two bifurcation surfaces and satisfies $\partial\mathcal D=B_1\cup B_2$.}
\label{fig:bifurcation}
\end{figure}

%\sout{Using the Killing equation $\nabla_{(\mu}\xi_{\nu)}=0$ and the Bianchi identity, one obtains}
%\begin{equation}
%   \xcancel{ \nabla_{\mu}\nabla_{\nu}\xi_{\rho}=R_{\rho\nu\mu\sigma}\xi^{\sigma},}
%\end{equation}
%\sout{where $R_{\rho\nu\mu\sigma}$ is the Riemann curvature tensor of the spacetime.}

To analyse the sign of surface gravity on $\mathcal H_{1}$ and $\mathcal H_{2}$, we introduce a future-directed ingoing auxiliary null vector $l^{\mu}$, chosen to point toward the black hole interior and normalized by $\xi^{\mu}l_{\mu}=-1$ on the Killing horizon. From Eq.~\eqref{SG}, the surface gravity is then given by
\begin{equation}
 \kappa=-\frac12l^\mu\nabla_\mu N^2.
 \label{eq:kappa-N2}
\end{equation}
Because the horizons are the simple zeros of $N^2$ and $N^2$ is positive inside static region $\mathcal{V}$, its inward derivative has opposite signs at the two crossings. Specifically, the sign changes of $N^2$ across two consecutive simple zeros give
\begin{equation}
 \kappa_1<0,
 \qquad
 \kappa_2>0,
 \label{eq:kappa-signs}
\end{equation}
%\sout{In the next section, we will show that, under the SEC, the sign of the surface gravity on the inner horizon leads to a contradiction.} 
The geometric analysis and conclusions above is independent of any gravitational theory. In the following sections, we show that the sign of surface gravity leads to a contradiction in both Einstein and Lovelock gravity under certain conditions.

\section{Einstein gravity: lapse equation and Komar charge}
\label{sec:einstein}

%\sout{We first review the lapse argument for the static Einstein gravity inner horizon constraint,}  
We first review the constraint on the number of inner Killing horizons derived from the lapse equation in Einstein gravity, and then develop an alternative derivation in terms of the Komar charge. 
%\sout{The two approaches are in fact equivalent, while the Komar formulation offers a complementary physical perspective in terms of gravitational charge.} 
The Komar charge provides an integral reformulation of the local lapse equation, and this relation is easy to be applied into other gravity theories. In the next section, we show that the Komar charge is precisely the gravitational Noether charge associated with the Killing vector in Einstein gravity, providing a natural starting point for its extension to Lovelock gravity.

Consider the trace-reversed form of Einstein's equation
\begin{equation}\label{eq:einstein-trace-reverse}
    R_{\mu\nu}=8\pi\left(T_{\mu\nu}-\frac{1}{d-1}Tg_{\mu\nu} \right),
\end{equation}
where $R_{\mu\nu}$ is $d+1$ dimensional Ricci tensor on the spacetime $\mathcal{M}$, $T_{\mu\nu}$ is the energy momentum tensor and $T=g^{\mu\nu}T_{\mu\nu}$ is its trace. The discussion below proceeds in three parts. We first derive the lapse equation and use it to recover the static inner horizon constraint. We then reformulate the same obstruction in terms of the Komar charge balance law, and finally show explicitly that the lapse equation and Komar  charge balance are two equivalent representations of the same geometric relation.

\subsection{lapse equation}\label{IIA}
For the static metric \eqref{metric}, staticity implies that the extrinsic curvature of the constant-t slices $\Sigma_{t}$ vanishes, i.e., $K_{ab}=0$. The ADM decomposition of Einstein's equation on $\mathcal D$ yields the following two independent equations:
\begin{equation}\label{Const1}
    ^{(d)}R=16\pi\rho
\end{equation}
and
\begin{equation}\label{Evolua1}
    ^{(d)}R_{ab}= N^{-1}D_aD_bN+8\pi\left[\mathcal{T}_{ab}+\frac{h_{ab}}{d-1}(\rho-\mathcal{T}) \right].
\end{equation}
Here $D_{a}$ and $^{(d)}R_{ab}$ are the induced covariant derivative operator associated with $h_{ab}$ and $d$-dimensional Ricci tensor on hypersurface $\Sigma_{t}$, respectively. $^{(d)}R=h^{ab}R_{ab}$ is the scalar curvature on $\Sigma_{t}$. $\rho=T_{\mu\nu}n^{\mu}n^{\nu}$ is the energy density, $\mathcal{T}_{ab}$ is the projection of stress tensor $T_{\mu\nu}$ on $\Sigma_{t}$ and $\mathcal{T}=h^{ab}\mathcal{T}_{ab}$ is its trace. Combining Eq.~\eqref{Const1} and Eq.~\eqref{Evolua1}, it can be obtained that
\begin{equation}\label{D2N}
    D^2N=8\pi N \left[\rho+\frac{T}{d-1} \right]=NR_{\mu\nu}n^{\mu}n^{\nu}.
\end{equation}
The SEC states that the bracket in Eq.~\eqref{eq:einstein-trace-reverse} satisfies $[T_{\mu\nu}-Tg_{\mu\nu}/(d-1)]v^{\mu}v^{\nu}\ge0$ for any timelike vector $v^{\mu}$. Hence, through Einstein's equation~\eqref{eq:einstein-trace-reverse}, the SEC implies the timelike convergence condition (TCC), $R_{\mu\nu}v^\mu v^\nu\geq 0$ for any timelike vector $v^\mu$ and
\begin{equation}
 D^2N^2=2N^2R_{\mu\nu}n^\mu n^\nu+2h^{ab}D_aN D_{b}N\geq0.
 \label{eq:N2-subharmonic}
\end{equation}
It shows that, on a compact region $\mathcal{D}$, $N^2$ is subharmonic and $N^2=0$ on the boundary $\partial\mathcal{D}$. The maximum principle then implies $N^2\leq0$, contradicting the fact that $N^2>0$ inside the $\mathcal{D}$. Hence, under the SEC, a static region $\mathcal V$ bounded by two consecutive nondegenerate inner Killing horizons cannot exist.

This immediately constrains the number of inner horizons. Indeed, if two or more nondegenerate inner Killing horizons existed on the same black hole interior branch, there would necessarily exist a static region bounded by two inner horizons. Therefore, under the SEC, there can be at most one nondegenerate inner Killing horizon inside the event horizon. This is precisely the compact static inner horizon constraint derived in Ref.~~\cite{yang2022horizons}. We next derive the same conclusion from a different perspective, using the Komar integral and the associated Komar charge balance law.

\subsection{Komar charge balance law}

For a Killing vector field $\xi^{\mu}$, the antisymmetric tensor $\nabla^{\mu} \xi^{\nu}$ defines a Komar charge on any closed 2-codimensional surface $S$~\cite{Komar:1958wp},
\begin{equation}
 \mathcal Q_{\rm EH}[S]
 =-\frac{1}{16\pi}\int_S \nabla^\mu\xi^\nu\,\dd S_{\mu \nu}.
 \label{eq:EH-charge}
\end{equation}
Here $\dd S_{\mu \nu}=\epsilon_{\mu \nu}\dd S$ is area $2$-form of surface $S$. This Komar charge coincides with the gravitational Noether charge in Einstein gravity introduced in the next section.

We next apply this Komar charge on the two inner horizons. As discussed in Sec.~\ref{sec:pocket}, the static slice $\mathcal{D}$ terminates at the bifurcation surfaces $B_1$ and $B_2$ associated with the two nondegenerate Killing horizons. On each bifurcation surface, $\nabla_a\xi_b=\kappa\eps_{ab}$, where $\eps_{\mu\nu}\eps^{\mu\nu}=-2$. Substitution into the Komar charge gives
\begin{equation}
 \mathcal Q_{\rm EH}[B]
 =\frac{\kappa}{2\pi}\frac{A_H}{4},
 \label{eq:EH-horizon-charge}
\end{equation}
where $A_{H}$ is the area of $d-1$ dimensional bifurcation surface $B$. Since the areas of regular horizon sections are strictly positive, the sign of surface gravity Eq.~\eqref{eq:kappa-signs} implies the Komar charges on these two horizon sections satisfy
\begin{equation}\label{Qsign}
    \begin{aligned}
         &\mathcal Q_{\rm EH}[B_{1}]=\frac{\kappa_{1}}{2\pi}\frac{A_{H_{1}}}{4}<0,\\
         &\mathcal Q_{\rm EH}[B_{2}]=\frac{\kappa_{2}}{2\pi}\frac{A_{H_{2}}}{4}>0
    \end{aligned}
\end{equation}

Applying the Stokes' theorem, the difference between the Komar charges can then be expressed as bulk integral over $\mathcal{D}$:
\begin{equation}\label{eq:EH-bulk-second-derivative}
     \mathcal Q_{\rm EH}[B_{1}]- \mathcal Q_{\rm EH}[B_{2}]=-\frac{1}{8\pi}\int_{\mathcal{D}}\nabla_{\nu}\nabla^\mu\xi^\nu\,\dd \Sigma_{\mu},
\end{equation}
where $\dd \Sigma_{\mu}=-n_{\mu}\dd \Sigma$ is the $d$-dimensional volume element on $\mathcal{D}$. The second derivative of a Killing vector is fixed by the spacetime curvature. Using the Killing equation $\nabla_{(\mu}\xi_{\nu)}=0$ together with the cyclic identity $R_{\mu[\nu\rho\sigma]}=0$, one has
\begin{equation}
    \nabla_\mu\nabla_\nu\xi_\rho
    =
    R_{\rho\nu\mu\sigma}\xi^\sigma .
    \label{eq:Killing-second-derivative}
\end{equation}
Combining with $\xi^{\mu}=Nn^{\mu}$ and substituting Eq.~\eqref{eq:Killing-second-derivative} into
Eq.~\eqref{eq:EH-bulk-second-derivative} yields
\begin{equation}\label{eq:EH-bulk-second-derivative2}
     \mathcal Q_{\rm EH}[B_{1}]- \mathcal Q_{\rm EH}[B_{2}]=\frac{1}{8\pi}\int_{\mathcal{D}}NR_{\mu\nu}n^{\mu}n^{\nu}\,\dd \Sigma,
\end{equation}
In static region $N>0$ and the SEC therefore imply $ \mathcal Q_{\rm EH}[B_{1}]- \mathcal Q_{\rm EH}[B_{2}]\ge0$, in contradiction with the sign of Eq.~\eqref{Qsign}. The assumed static region $\mathcal{V}$ bounded by two consecutive nondegenerate inner Killing horizons is therefore impossible. Thus, the Komar argument reproduces the same inner horizon constraint: under the SEC, a static region bounded by two nondegenerate inner Killing horizons cannot exist, and hence there can be at most one nondegenerate inner Killing horizon inside the event horizon.

More importantly, the Komar charge balance law provides a different physical interpretation of this obstruction. The surface gravities fix the signs of the Komar charges on the two horizon boundaries and hence fix the sign of their difference, while the SEC fixes the sign of the bulk integral relating them. For two inner horizons, these two sign requirements are incompatible.

\subsection{Relation between the lapse and Komar arguments}

The Komar construction above is not an independent trick unrelated to the lapse argument. In a static spacetime, the two approaches are different representations of the same geometric identity. The lapse formulation keeps the relation in a local differential form, whereas the Komar formulation converts it into a bulk--boundary balance.

To see this explicitly, recall that in the static region $\xi^\mu=Nn^\mu$. Using the orientation convention introduced in Eq.~\eqref{eq:binormal}, the boundary element on
$\partial\mathcal D$ is $\dd S{\mu\nu}=2n_{[\mu}s_{\nu]}\,\dd S$. Since $\nabla^\mu\xi^\nu$ is antisymmetric, its contraction with the boundary element gives
\begin{equation}
\begin{aligned}
    \nabla^\mu\xi^\nu \dd S{\mu\nu}
    &=
    -2n_\nu s_\mu\nabla^\mu\xi^\nu\,\dd S
    \\
    &=
    -N^{-1} s^\mu \nabla_\mu (\xi^2)\,\dd S \\
    &=2s^\mu \partial_\mu N\,\dd S.
\end{aligned}
\end{equation}
Integrating over the boundary and applying the Stokes' theorem therefore yields
\begin{equation}
    \int_{\partial\mathcal D}
    \nabla^\mu\xi^\nu\,\dd S{\mu\nu}
    =
    2\int_{\partial\mathcal D}
    s^\mu D_\mu N\,\dd S
    =
    2\int_{\mathcal D}
    D^2N\,\dd \Sigma .
    \label{eq:komar-lapse-equivalence}
\end{equation}
It shows that the bulk representation of the Komar charge is precisely obtained by integrating the lapse equation~\eqref{D2N}. The lapse equation provides a pointwise differential description in the static region, whereas the Komar formulation converts its integral over the static region into a Komar charge balance.

It is important that Eq.~\eqref{eq:komar-lapse-equivalence} is geometric: its derivation uses Killing symmetry and the Stokes' theorem, but not specific gravitational field equations. However, the expression for $D^2 N$ in terms of the matter and curvature sources depends on the specific gravitational theory. In a modified gravity theory, one may formally rewrite the field equations in an Einstein-like form by absorbing the higher-curvature contributions into an effective energy-momentum tensor,
\begin{equation}
    G_{\mu\nu}=
    R_{\mu\nu}-\frac{1}{2}g_{\mu\nu}R
    =
    8\pi \,T^{\rm eff}_{\mu\nu}.
\end{equation}
If the effective energy-momentum tensor, which contains the higher-curvature contributions, satisfies the corresponding effective SEC, then the same lapse argument can be applied and leads formally to the same static inner horizon constraint as in Einstein gravity. This conclusion, however, follows from an Einstein-like rearrangement of the field equation: all curvature terms other than the Einstein tensor $G_{\mu\nu}$ are simply moved to the effective stress-tensor side. Although this rewriting is mathematically legitimate, it obscures the intrinsic bulk-boundary structure of the higher-curvature theory. In particular, the Komar form $\nabla^{\mu}\xi^{\nu}$, which represents the gravitational Noether potential in Einstein gravity, is no longer the Noether potential associated with the higher-curvature action.

A direct ADM decomposition of a higher-curvature theory is also possible in principle, but the resulting constraint and evolution equations generally involve considerably more complicated combinations of intrinsic and extrinsic curvatures. In this form, the simple scalar equation for the lapse that underlies the proof in Einstein gravity is no longer manifest. The gravitational Noether charge formulation provides a more intrinsic alternative. Instead of moving the higher-curvature terms to the effective stress-tensor side of the field equations, one can work directly with the gravitational Noether charge determined by the gravitational action and examine the corresponding bulk-boundary balance. This is the route we pursue below.

%%%%%%%%%%%%%%%%%%%%%%%%%%%%%%%%%%%%%%%%%%%%%%%%%%%%%%%%%%%%%%%%%%%%%

\section{From Komar charge to the gravitational Noether Charge}
\label{sec:wald-charge}

To extend the boundary-charge formulation beyond Einstein gravity, we first identify the gravitational charge associated with diffeomorphism invariance in a general theory. We therefore begin with a diffeomorphism invariant action and review the corresponding Iyer-Wald Noether current, potential, and charge~\cite{Wald:1993nt,Iyer:1994ys}. We then specialize to Lovelock gravity, for which the Noether potential takes a particularly simple form.

Consider a diffeomorphism invariant theory whose gravitational Lagrangian depends algebraically on the metric $g_{\mu\nu}$ and the Riemann tensor $R_{\mu\nu\rho\sigma}$, but not on derivatives of the curvature,
\begin{equation}
    I[g,\psi]
    =
    \frac{1}{16\pi}
    \int_{\mathcal M}
    \dd^D x\,\sqrt{-g}\,
    L_{\rm grav}
    \bigl(g_{\mu\nu},R_{\mu\nu\rho\sigma}\bigr)
    +
    I_{\rm m}[g,\psi],
    \label{eq:general-action}
\end{equation}
where $D=d+1$ is the spacetime dimension, $I_{\mathrm m}$ is the action of the matter fields and $\psi$ collectively denotes the matter fields. It is convenient to introduce the curvature response tensor
\begin{equation}
    P^{\mu\nu\rho\sigma}
    :=
    \frac{\partial L_{\rm grav}}
    {\partial R_{\mu\nu\rho\sigma}} .
    \label{eq:P-definition}
\end{equation}
Since the Riemann tensor possesses its usual algebraic symmetries, $ P^{\mu\nu\rho\sigma}$ inherits the same symmetries,
\begin{equation}
    P^{\mu\nu\rho\sigma}
    =
    -P^{\nu\mu\rho\sigma}
    =
    -P^{\mu\nu\sigma\rho}
    =
    P^{\rho\sigma\mu\nu},
    \quad
    P^{\mu[\nu\rho\sigma]}=0 .
    \label{eq:P-symmetry}
\end{equation}
To obtain the field equations and the gravitational Noether charge, we first consider a variation of the metric. The variation of the gravitational Lagrangian separates naturally into a bulk term, which determines the equations of motion, and a total-divergence term, which will enter the construction of the Noether current:
\begin{equation}
    \delta\!\left(\sqrt{-g}\,L_{\rm grav}\right)
    =
    \sqrt{-g}\,
    E_{\mu\nu}\,\delta g^{\mu\nu}
    +
    \sqrt{-g}\,
    \nabla_\mu\Theta^\mu(g,\delta g),
    \label{eq:variation-L}
\end{equation}
Here $\Theta^\mu(g,\delta g)$ is the symplectic-potential current. Including the variation of the matter action, the field equation is
\begin{equation}
\begin{aligned}
    E_{\mu\nu}
    &=
    P_{\mu}{}^{\rho\sigma\lambda}
    R_{\nu\rho\sigma\lambda}
    -
    \frac12 L_{\rm grav}g_{\mu\nu}
    -
    2\nabla^\rho\nabla^\sigma
    P_{\mu\rho\sigma\nu}
    \\&=
    8\pi \,T_{\mu\nu}.
    \end{aligned}
    \label{eq:general-field-equation}
\end{equation}
For an infinitesimal diffeomorphism generated by a vector field $\xi^\mu$ (it need not be a Killing vector field), the metric variation is $\delta_\xi g_{\mu\nu}=\mathcal L_\xi g_{\mu\nu}$. From Eq.~\eqref{eq:variation-L}, diffeomorphism invariance then implies the associated gravitational Noether current
\begin{equation}\label{current}
    J^\mu_\xi
    =\Theta^\mu(g,\mathcal L_\xi g)
    -
    \xi^\mu L_{\rm grav}
\end{equation}
is conserved on shell, i.e., $\nabla_\mu J^\mu_\xi=0$. The Poincar\'e lemma then implies that, locally, this current can be written as the divergence of an antisymmetric Noether potential $J^\mu_\xi=2\nabla_\nu q_\xi^{\nu\mu}$. For a diffeomorphism invariant gravitational Lagrangian $L_{\rm grav}\bigl(g_{\mu\nu},R_{\mu\nu\rho\sigma}\bigr)$ containing no derivatives of the curvature, the corresponding Noether potential can be written as (see Appendix~\ref{app:noether-potential} for a derivation)
\begin{equation}
    q_\xi^{\mu\nu}
    =
        P^{\mu\nu\rho\sigma}\nabla_\rho\xi_\sigma
        +
        2\xi_\rho\nabla_\sigma
        P^{\mu\nu\rho\sigma}.
    \label{eq:general-IW-potential}
\end{equation}
The associated gravitational Noether charge on a closed
codimension-two surface $S$ is then
\begin{equation}
    \mathcal Q_\xi[S]
    :=
    -\frac{1}{16\pi}
    \int_S
    q_\xi^{\mu\nu}\,
    \dd S_{\mu\nu}.
    \label{eq:general-IW-charge}
\end{equation}
%\sout{In contrast to an Einstein-like rewriting in terms of an effective energy-momentum tensor,} 
Rather than absorbing the higher-curvature contributions into an effective energy-momentum tensor, the gravitational charge used here is derived directly from the original higher-curvature action. 

We next specialize this charge to Lovelock gravity~\cite{Padmanabhan:2013xyr,Chakraborty:2014joa}. Its gravitational Lagrangian is a finite sum
\begin{equation}
    L_{\rm grav}
    =
    \sum_{m=0}^{K} c_m L_m ,
    \label{eq:lovelock-L}
\end{equation}
where $c_m$ are the corresponding coupling constants. The $m$-th Lovelock density is defined by
\begin{equation}
    L_m
    =
    \frac{1}{2^m}
    \delta^{
        \mu_1\nu_1\cdots\mu_m\nu_m
    }_{
        \rho_1\sigma_1\cdots\rho_m\sigma_m
    }
    \prod_{j=1}^{m}
    R^{\rho_j\sigma_j}{}_{\mu_j\nu_j}.
    \label{eq:lovelock-density}
\end{equation}
Here $\delta$ denotes the generalized Kronecker delta, which is totally antisymmetric in both sets of indices. Owing to this complete antisymmetrization, $L_m$ vanishes identically for $D<2m$. In the critical dimension $D=2m$, its integral is proportional to the Euler characteristic and therefore does not contribute to the local metric field equations. Hence, for a fixed spacetime dimension, only finitely many Lovelock terms contribute nontrivially to the dynamics. The first few terms have familiar interpretations:
\begin{itemize}
    \item {$m=0$:} $L_{0}=1$, which gives the cosmological constant term $c_{0}\int\dd^{D}x\sqrt{-g}$.
    \item {$m=1$:} $L_{1}=R$, the Einstein-Hilbert density.
    \item {$m=2$:} $L_{2}=R_{\mu\nu\rho\sigma}R^{\mu\nu\rho\sigma}-4R_{\mu\nu}R^{\mu\nu}+R^2$, the Gauss-Bonnet density.
    \item {$m\ge3$:} the higher-order Lovelock densities.
\end{itemize}

The complete antisymmetrization of the generalized Kronecker delta, together with the differential Bianchi identity for the Riemann tensor, implies
\begin{equation}
    \nabla_\mu P^{\mu\nu\rho\sigma}=0 .
    \label{eq:lovelock-divP}
\end{equation}
This property is special: it removes the $\nabla^\rho\nabla^\sigma P_{\mu\rho\sigma\nu}$ term from the field equation~\eqref{eq:general-field-equation} and simultaneously removes the $\xi_\rho\nabla_\sigma P^{\mu\nu\rho\sigma}$ term from the Noether potential~\eqref{eq:general-IW-potential}. Define the Ricci-like tensor $\mathcal{R}_{\mu\nu}:= P_{\mu}{}^{\rho\sigma\lambda}R_{\nu\rho\sigma\lambda}$. The Lovelock field equation therefore reduces to
\begin{equation}
    \mathcal{R}_{\mu\nu}
    -
    \frac{1}{2}L_{\rm grav}g_{\mu\nu}
    =
    8\pi \,T_{\mu\nu},
    \label{eq:lovelock-field-equation}
\end{equation}
whereas the gravitational Noether potential takes the simple
Komar-like form
\begin{equation}
    q_\xi^{\mu\nu}
    =
    P^{\mu\nu\rho\sigma}
    \nabla_\rho\xi_\sigma .
    \label{eq:lovelock-noether-potential}
\end{equation}
Consequently,
\begin{equation}
    \mathcal Q_\xi[S]
    =
    -\frac{1}{16\pi}
    \int_S
    P^{\mu\nu\rho\sigma}
    \nabla_\rho\xi_\sigma\,
    \dd S_{\mu\nu}.
    \label{eq:lovelock-Noether charge}
\end{equation}
This is the precise sense in which the Komar construction extends naturally within the Lovelock class: the Komar integrand $\nabla^\mu\xi^\nu$ in Einstein gravity is replaced by the curvature-response-weighted quantity
$P^{\mu\nu\rho\sigma}\nabla_\rho\xi_\sigma$, with $P^{\mu\nu\rho\sigma}$ determined uniquely by the gravitational Lagrangian.

The Einstein gravity result is recovered immediately from the first
Lovelock order. For the Einstein-Hilbert density $L_1=R$,
\begin{equation}
    P_{\rm EH}^{\mu\nu\rho\sigma}
    =
    \frac12
    \left(
        g^{\mu\rho}g^{\nu\sigma}
        -
        g^{\mu\sigma}g^{\nu\rho}
    \right).
    \label{eq:P-EH}
\end{equation}
For the Killing field $\xi^{\mu}$,
\begin{equation}
    P_{\rm EH}^{\mu\nu\rho\sigma}
    \nabla_\rho\xi_\sigma
    =
    \nabla^\mu\xi^\nu ,
\end{equation}
and Eq.~\eqref{eq:lovelock-Noether charge} becomes
\begin{equation}
    \mathcal Q_\xi[S]
    =
    -\frac{1}{16\pi }
    \int_S
    \nabla^\mu\xi^\nu\,
    \dd S_{\mu\nu},
\end{equation}
which is precisely the Komar charge used in the previous section.
The Einstein Komar charge is therefore not a separate construction
from the Wald formalism: it is the $m=1$ Lovelock realization of the
gravitational Noether charge.

%%%%%%%%%%%%%%%%%%%%%%%%%%%%%%%%%%%%%%%%%%%%%%%%%%%%%%%%%%%%%%%%%%%%%%%

\section{Wald Entropy and the Noether charge Balance Law}

In this section, we establish the relation between horizon Noether charges and the bulk curvature in Lovelock gravity. We first investigate the gravitational Noether charge on a bifurcation surface and show that it is given by the product of the surface gravity $\kappa$ and the Wald entropy $S_{\mathrm W}$. We then use the divergence identity of the Lovelock Noether potential together with Stokes' theorem to relate the difference between the two horizon charges to a bulk integral of the Ricci-like tensor $\mathcal{R}_{\mu\nu}$. This yields the Lovelock generalization of the Einstein Komar balance law.

\subsection{Wald entropy}
The Noether charge acquires a particularly simple and physically transparent form when evaluated on a bifurcation surface $B$. To make this relation explicit, we first introduce the Wald entropy associated with $B$, defined by~\cite{Wald:1993nt,Iyer:1994ys}
\begin{equation}
    S_{\rm W}[B]
    =
    -\frac{1}{8}
    \int_B
    P^{\mu\nu\rho\sigma}
    \epsilon_{\mu\nu}\epsilon_{\rho\sigma}\,
    \dd S .
    \label{eq:wald-entropy}
\end{equation}
Using $\dd\Sigma_{\mu\nu}=\epsilon_{\mu\nu}\dd S$, together with the expression for the surface gravity on the bifurcation surface given in Eq.~\eqref{eq:bifurcation-kappa}, the Noether charge integrand Eq.~\eqref{eq:lovelock-Noether charge} reduces to the same binormal contraction appearing in the Wald entropy. Hence
\begin{align}
    \mathcal Q_\xi[B]
    &=
    -\frac{\kappa}{16\pi}
    \int_B
    P^{\mu\nu\rho\sigma}
    \epsilon_{\mu\nu}\epsilon_{\rho\sigma}\,
    \dd A
    \nonumber\\
    &=
    \frac{\kappa}{2\pi}\,S_{\rm W}[B].
    \label{eq:charge-kappa-entropy}
\end{align}
Thus the gravitational Noether charge on the horizon section is the product of the surface gravity factor and the Wald entropy. For the Lovelock Lagrangian $L_{\rm grav}=\sum_{m=0}^{K}c_m L_m$, Eq.~\eqref{eq:wald-entropy} can be written entirely in terms of the
intrinsic geometry of the bifurcation surface~\cite{Jacobson:1993xs}:
\begin{equation}
    S_{\rm W}[B]
    =
    \frac{1}{4}
    \sum_{m=1}^{K}
    m c_m
    \int_B
    \widetilde L_{m-1}\,
    \dd S .
    \label{eq:lovelock-wald-entropy}
\end{equation}
Here $\widetilde L_{m-1}$ denotes the $(m-1)$-th Lovelock density constructed from the intrinsic $(d-1)$-dimensional Riemann tensor of the induced metric on $B$. The $m=1$ contribution is simply $S_{\rm W}^{(1)}=A_{H}/4$, so the Bekenstein-Hawking entropy is the first nontrivial Lovelock contribution to the Wald entropy. Applying Eq.~\eqref{eq:charge-kappa-entropy} to the two bifurcation surfaces $B_1$ and $B_2$, we therefore have
\begin{equation}
    \mathcal Q_\xi[B_i]
    =
    \frac{\kappa_i}{2\pi}
    S_{\rm W}[B_i],
    \qquad i=1,2 .
    \label{eq:two-horizon-charges}
\end{equation}

\subsection{Noether charge balance law}

We now compare the Noether charges on the two bifurcation surfaces $B_1$ and $B_2$. Taking their difference and applying Stokes' theorem to the completed static hypersurface $\mathcal{D}$ gives
\begin{equation}
     \mathcal Q_\xi[B_1]- \mathcal Q_\xi[B_2]=\frac{1}{8\pi}
    \int_\mathcal D
    n_{\mu}\nabla_{\nu}(P^{\mu\nu\rho\sigma}
    \nabla_\rho\xi_\sigma)\,
    \dd \Sigma.
\end{equation}
Using the Lovelock identity $\nabla_{\mu}P^{\mu\nu\rho\sigma}=0$, together with the Killing identity Eq.~\eqref{eq:Killing-second-derivative}, we obtain
\begin{equation}
\begin{aligned}
    \nabla_\mu
    \left(
        P^{\mu\nu\rho\sigma}
        \nabla_\rho\xi_\sigma
    \right)
    &=
    P^{\mu\nu\rho\sigma}
    \nabla_\mu\nabla_\rho\xi_\sigma
    \nonumber\\
    &=
    -\mathcal R^\nu{}_{\lambda}\xi^\lambda .
    \label{eq:lovelock-killing-divergence}
\end{aligned}
\end{equation}
Therefore, the Noether charge balance law
\begin{equation}
    \mathcal Q_\xi[B_1]
    -
    \mathcal Q_\xi[B_2]
    =
    \frac{1}{8\pi}
    \int_{\mathcal D}
    N\,
    \mathcal R_{\mu\nu}
    n^\mu n^\nu\,
    \dd \Sigma .
    \label{eq:lovelock-charge-balance}
\end{equation}
Here the static decomposition $\xi^\mu=Nn^\mu$ is used. Finally, substituting the charge expression on horizon sections Eq.~\eqref{eq:two-horizon-charges} into Eq.~\eqref{eq:lovelock-charge-balance} gives
\begin{equation}
    \frac{\kappa_1}{2\pi}
    S_{\rm W}[B_1]
    -
    \frac{\kappa_2}{2\pi}
    S_{\rm W}[B_2]
    =
    \frac{1}{8\pi}
    \int_{\mathcal D}
    N\,
    \mathcal R_{\mu\nu}
    n^\mu n^\nu\,
    \dd \Sigma .
    \label{eq:wald-balance-law}
\end{equation}
This is the Lovelock counterpart of the Einstein Komar charge balance law. It relates the horizon data on the two boundaries to an integral of the Ricci-like curvature over the static region.

%%%%%%%%%%%%%%%%%%%%%%%%%%%%%%%%%%%%%%%%%%%%%%%%%%%%%%%%%%%%%%%%%%%%%%
\section{TCC-like Curvature Condition and the inner horizon Constraint}
\label{sec:curvature-condition}

The Noether charge balance law derived in the previous section shows that the Ricci-like tensor $\mathcal R_{\mu\nu}$ is the curvature quantity governing the bulk contribution between the two horizon boundaries. In Einstein gravity, the corresponding sign condition is the TCC, which is implied by the SEC. This suggests introducing a Lovelock analogue in which the Ricci tensor is replaced by the Ricci-like tensor selected by the Noether charge balance law.
\begin{definition}
In Lovelock gravity, a spacetime is said to satisfy the TCC-like curvature condition if
\begin{equation}
    \mathcal R_{\mu\nu}v^\mu v^\nu \ge 0
    \label{eq:TCC-like-condition}
\end{equation}
for every timelike vector $v^\mu$.
\end{definition}
Unlike the Einstein TCC, the corresponding quantity $\mathcal R_{\mu\nu}v^\mu v^\nu$ in Lovelock gravity receives contributions not only from the matter energy-momentum tensor but also from the higher-curvature terms. This can be seen explicitly from Eq.~\eqref{eq:lovelock-field-equation}. Since the $m$-th Lovelock density $L_m$ is homogeneous of degree $m$ in the Riemann tensor, the Euler's theorem for homogeneous functions gives
\begin{equation}
    P_{(m)}^{\mu\nu\rho\sigma}
    R_{\mu\nu\rho\sigma}
    =\frac{\partial L_{m}}{\partial R_{\mu\nu\rho\sigma}}R_{\mu\nu\rho\sigma}=
    m L_m .
\end{equation}
Here $P_{(m)}^{\mu\nu\rho\sigma}$ denotes the curvature response tensor associated with the $m$-th Lovelock density $L_m$. The trace of the Ricci-like tensor therefore is $\mathcal R:=g^{\mu\nu}\mathcal R_{\mu\nu}=\sum_{m=0}^{K} m c_m L_m$. Taking the trace of the Lovelock field equations Eq.~\eqref{eq:lovelock-field-equation} then yields
\begin{equation}
    \sum_{m=0}^{K}
    c_m
    \left(
        m-\frac{d+1}{2}
    \right)
    L_m
    =
    8\pi  T .
    \label{eq:lovelock-trace}
\end{equation}
Combining $L_{\rm grav}=\sum_{m=0}^{K}c_mL_m$ with Eq.~\eqref{eq:lovelock-trace}, the Lovelock field equations Eq.~\eqref{eq:lovelock-field-equation} can be rewritten as
\begin{equation}
    \mathcal R_{\mu\nu}
    =
    8\pi
    \left(
        T_{\mu\nu}
        -\frac{1}{d-1}g_{\mu\nu}T
    \right)
    +
    \frac{g_{\mu\nu}}{d-1}
    \sum_{m=0}^{K}
    (m-1)c_m L_m .
    \label{eq:ricci-like-decomposition}
\end{equation}
This expression may be viewed as the Lovelock counterpart of the trace-reversed Einstein equation. Unlike in Einstein gravity, its right-hand side contains not only the matter combination but also an explicit contribution from the higher-curvature terms. The TCC-like curvature
condition is therefore not a simple matter energy condition.

With the TCC-like curvature condition specified, we can now establish the inner horizon constraint in Lovelock gravity.
\begin{theorem}
Consider a static spacetime in Lovelock gravity with Killing vector $\xi^\mu$ and regular Killing horizons with compact spatial cross sections. If
\begin{enumerate}
    \item[(a)] the TCC-like curvature condition holds;
    \item[(b)] the Wald entropy of every Killing horizon associated
    with $\xi^\mu$ is positive.
\end{enumerate}
Then, inside the event horizon, there exists at most one nondegenerate Killing horizon associated with $\xi^\mu$.
\end{theorem}

\textit{Proof:} To the contrary, suppose that there exist two consecutive nondegenerate inner Killing horizons associated with the Killing vector $\xi^\mu$. By Eq.~\eqref{eq:kappa-signs}, the surface gravities of the two consecutive nondegenerate Killing horizons satisfy $\kappa_1<0<\kappa_{2}$. Since the Wald entropy is positive on both horizons, Eq.~\eqref{eq:two-horizon-charges} then gives
\begin{equation}
    \mathcal Q_\xi[B_1]-\mathcal Q_\xi[B_2]<0.
    \label{eq:horizon-charge-negative}
\end{equation}
On the other hand, since the unit normal $n^{\mu}$ to $\mathcal{D}$ is timelike, the TCC-like curvature condition implies $\mathcal R_{\mu\nu}n^\mu n^\nu\ge0$. Together with $N>0$ in the interior of $\mathcal D$, the Noether charge balance law \eqref{eq:lovelock-charge-balance} then gives
\begin{equation}
    \mathcal Q_\xi[B_1]-\mathcal Q_\xi[B_2]\ge0 .
    \label{eq:charge-monotonicity}
\end{equation}
Eqs.~\eqref{eq:horizon-charge-negative} and \eqref{eq:charge-monotonicity} are incompatible. Therefore, the static region $\mathcal V$ bounded by two consecutive nondegenerate inner Killing horizons cannot exist. Hence, there can be at most one nondegenerate inner Killing horizon associated with $\xi^\mu$ inside the event horizon. $\square$

A brief remark on the positive entropy assumption of the theorem is in order. On the one hand, the condition $S_{\rm W}>0$ is essential to the proof, since it allows the sign of the horizon Noether charge to be determined from the surface gravity. On the other hand, the positivity assumption is also motivated by the thermodynamic interpretation of Wald entropy. Moreover, Wald entropy has been shown to obey the second law during quasi-stationary physical processes in Lovelock gravity, provided that the matter satisfies the null energy condition (NEC) \cite{Kolekar:2012tq}. The positivity of Wald entropy is also a natural requirement from the perspective of black hole thermodynamics. Nevertheless, unlike in Einstein gravity, where $S_{\rm W}=A_{H}/4>0$, the positivity of the Wald entropy is not automatic in Lovelock gravity because of higher-curvature contributions. For example, in Einstein-Gauss-Bonnet gravity with coupling constant $\alpha$, the Wald entropy is given by
\begin{equation}
S_{\rm W}[B]=\frac14\int_B(1+2\alpha\widetilde R)\,dS,
\end{equation}
where $\widetilde R$ is the intrinsic scalar curvature of the horizon cross section $B$. For $\alpha\geq0$, a nonnegative $\widetilde R$ guarantees positive Wald entropy, whereas negative intrinsic curvature may lead to negative entropy contributions.
%\sout{We therefore take $S_{\rm W}>0$ as a criterion of thermodynamic admissibility, selecting the branch of Lovelock black holes for which the horizon entropy retains its standard positive thermodynamic interpretation.} 
From a statistical mechanical perspective, black hole entropy is expected to be related to the logarithm of the number of underlying microscopic states \cite{Wald:1999vt}. In the conventional microstate counting interpretation, the number of accessible states is at least one, so the entropy is nonnegative. We therefore impose $S_{\rm W}>0$ as an additional condition on the Killing horizons considered here.
%%%%%%%%%%%%%%%%%%%%%%%%%%%%%%%%%%%%%%%%%%%%%%%%%%%%%%%%%%%%%%%

\section{Discussion}
\label{sec:discussion}
In this work, we show that the constraint on inner horizons inside static black holes can be derived using a Komar integral method. This provides an alternative formulation equivalent to the lapse equation approach developed in previous work. The Komar integral method has an important advantage. It recasts the local lapse equation as a bulk-boundary balance of gravitational Noether charge. For the two putative consecutive inner horizons, the TCC, which is guaranteed by the SEC through Einstein's equations, implies the corresponding bulk integral to be nonnegative, while the boundary Komar charges require the charge difference to be negative. These two sign requirements are incompatible. Since the Komar charge is precisely the Einstein gravity realization of the gravitational Noether charge, the Komar integral method admits a natural extension to higher-curvature gravity through the corresponding gravitational Noether charge. This approach has a particular advantage in higher-curvature theories, where the ADM decomposition of the field equations generally leads to a considerably more complicated lapse equation. In contrast, the gravitational Noether charge in Lovelock gravity retains a relatively simple covariant form. And this approach avoids rewriting the higher-curvature contributions as an effective stress-energy tensor. We then examine the corresponding gravitational Noether charge balance law in Lovelock gravity. In this balance law, the bulk contribution is governed by the Ricci-like tensor $\mathcal R_{\mu\nu}$, while the horizon Noether charge is determined by the surface gravity and the Wald entropy. 

Motivated by the ordinary TCC in Einstein gravity, we define a TCC-like curvature condition in Lovelock gravity. Together with the positivity of the Wald entropy, the gravitational Noether-charge method then shows that a black hole satisfying this TCC-like condition can contain at most one nondegenerate inner Killing horizon.

Our result reveals a deeper bulk-boundary structure behind the inner horizon constraint. The TCC-like curvature condition and the positivity of the Wald entropy are not two unrelated assumptions. Both are governed by the same curvature response tensor $P^{\mu\nu\rho\sigma}$. Its contraction with the spacetime curvature determines the Ricci-like tensor $\mathcal R_{\mu\nu}$, which controls the bulk contribution. Its contraction with the horizon binormal determines the Wald entropy, which controls the horizon contribution. The inner horizon constraint can therefore be understood as a compatibility requirement between the bulk contribution determined by the gravitational response and the horizon charge determined by thermodynamic data. This offers a new perspective on inner horizon constraints and suggests that the compatibility between gravitational response in the bulk and horizon thermodynamics may provide a useful framework for further investigations.

The present analysis is restricted to static spacetimes with regular nondegenerate Killing horizons and to Lovelock gravity, for which $\nabla_\mu P^{\mu\nu\rho\sigma}=0$. Several extensions deserve further investigation. One challenge concerns degenerate horizons, which require a separate treatment, since $\kappa=0$ removes the boundary charge sign information, and the bifurcation surface construction is no longer directly applicable. %\sout{Rotation modifies the bulk contribution to the charge balance law. In Einstein gravity, the rotational correction is a manifestly nonnegative square and hence preserves the required sign. In higher-curvature gravity, the analogous correction is not generally sign definite, making the extension to stationary black holes subtle.} 
Another arises from rotation, which modifies the gravitational Noether charge balance law. As shown in Appendix~\ref{app:rotation}, the horizon generators become $\chi_i^\mu=\xi^\mu+\Omega_i\psi^\mu$, and the bulk integral acquires an additional rotational contribution involving the gradient of the frame-dragging function $\omega$. In Einstein gravity, this rotational correction reduces to a manifestly nonnegative quadratic term, so the gravitational Noether charge method continues to exclude the existence of two or more nondegenerate inner Killing horizons in stationary axisymmetric spacetimes under the TCC. In Lovelock gravity, however, the rotational contribution involves additional higher-curvature terms whose sign is not generally definite. Consequently, the TCC-like curvature condition and the positivity of Wald entropy alone are insufficient to establish the corresponding inner horizon constraint for rotating black holes. And extending the analysis beyond Lovelock gravity presents further difficulties. More general higher-curvature theories with $\nabla_\mu P^{\mu\nu\rho\sigma}\neq0$ contain additional terms in both the Noether potential and its divergence, so the charge balance law is modified.

More broadly, our Noether charge based viewpoint is closely related to several active developments in gravitational theory. In Lovelock gravity, connections between bulk-surface decompositions of the gravitational action and Wald entropy have been explored~\cite{Kolekar:2010dm}. Recent studies~\cite{Hollands:2024vbe,Golshani:2024fry} of covariant phase space and dynamical black hole entropy have further developed the role of gravitational surface charges in black hole thermodynamics, especially in dynamical and nonstationary settings. And ongoing work~\cite{Wall:2024lbd} on higher-curvature gravity has highlighted the nontrivial status of horizon entropy, including its second law properties and possible corrections beyond the stationary Wald expression. Our result adds another perspective: gravitational Noether charges and horizon entropy can also constrain the allowed horizon structure of spacetime. From this viewpoint, the inner horizon problem provides a setting in which gravitational charges, higher-curvature dynamics, and horizon thermodynamics meet directly.

\section*{ACKNOWLEDGMENTS}
This work is supported by the Natural Science Foundation of China under Grant No.12375051 and Tianjin University Self-Innovation Fund Extreme Basic Research Project Grant No. 2025XJ22-0014 and 2025XJ21-0007.

\appendix

\section{Derivation of the gravitational Noether potential}
\label{app:noether-potential}

In this appendix, we derive the gravitational Noether potential used
in Sec.~\ref{sec:wald-charge}. We use the definitions and conventions introduced in the
main text.

For a diffeomorphism invariant gravitational Lagrangian
$L_{\rm grav}(g_{\mu\nu},R_{\mu\nu\alpha\beta})$, the metric variation
takes the form
\begin{equation}
\delta\!\left(\sqrt{-g}L_{\rm grav}\right)
=
\sqrt{-g}\,E_{\mu\nu}\delta g^{\mu\nu}
+
\sqrt{-g}\,\nabla_\mu\delta v^\mu .
\label{app:variation}
\end{equation}
The corresponding boundary term can be written as
\cite{Chakraborty:2014joa}
\begin{equation}
\delta v^\mu
=
2P^{\beta\nu\mu\alpha}
\nabla_\nu\delta g_{\alpha\beta}
-
2\delta g_{\alpha\beta}
\nabla_\nu P^{\beta\mu\nu\alpha}.
\label{app:boundary}
\end{equation}
Consider an infinitesimal diffeomorphism generated by a vector field
$\xi^\mu$. The metric variation is
\begin{equation}
\delta_\xi g_{\mu\nu}
=
\mathcal L_\xi g_{\mu\nu}
=
\nabla_\mu\xi_\nu+\nabla_\nu\xi_\mu .
\label{app:diffeo}
\end{equation}
Substituting Eq.~\eqref{app:diffeo} into
Eq.~\eqref{app:boundary}, we obtain
\begin{align}
\delta_\xi v^\mu
={}&
2P^{\beta\nu\mu\alpha}
\nabla_\nu
\left(
\nabla_\alpha\xi_\beta+\nabla_\beta\xi_\alpha
\right)
\nonumber\\
&-
2\left(
\nabla_\alpha\xi_\beta+\nabla_\beta\xi_\alpha
\right)
\nabla_\nu P^{\beta\mu\nu\alpha}.
\label{app:deltav1}
\end{align}
Using the algebraic symmetries and cyclic identity
$P^{\mu[\nu\alpha\beta]}=0$, together with the commutator $[\nabla_\mu,\nabla_\nu]\xi_\alpha=R_{\alpha\beta\mu\nu}\xi^\beta$,
Eq.~\eqref{app:deltav1} can be rearranged as
\begin{align}
\delta_\xi v^\mu
={}&
2\mathcal R^\mu{}_\nu\xi^\nu
-
4\xi^\nu
\nabla_\alpha\nabla_\beta
P^\mu{}_{\alpha\beta\nu}
\nonumber\\
&-
2\nabla_\nu
\left(
P^{\mu\nu\alpha\beta}\nabla_\alpha\xi_\beta
+
2\xi_\alpha\nabla_\beta
P^{\mu\nu\alpha\beta}
\right),
\label{app:deltav2}
\end{align}
where $\mathcal R^\mu{}_\nu=P^{\mu\alpha\beta\lambda}R_{\nu\alpha\beta\lambda}$. The first two terms in Eq.~\eqref{app:deltav2} are precisely the
combination appearing in the gravitational equation-of-motion tensor.
Indeed, using
\begin{equation}
E^\mu{}_\nu
=
\mathcal R^\mu{}_\nu
-
\frac12 L_{\rm grav}\delta^\mu{}_\nu
-
2\nabla_\alpha\nabla_\beta
P^\mu{}_{\alpha\beta\nu},
\label{app:E}
\end{equation}
Equation~\eqref{app:deltav2} therefore becomes
\begin{align}
\delta_\xi v^\mu
={}&
2E^\mu{}_\nu\xi^\nu
+
L_{\rm grav}\xi^\mu
\nonumber\\
&-
2\nabla_\nu
\left(
P^{\mu\nu\alpha\beta}\nabla_\alpha\xi_\beta
+
2\xi_\alpha\nabla_\beta
P^{\mu\nu\alpha\beta}
\right).
\label{app:deltav3}
\end{align}
From Eq.~\eqref{current}, the diffeomorphism Noether current is defined by
\begin{equation}
J^\mu_\xi
=
\delta_\xi v^\mu
-
\xi^\mu L_{\rm grav}.
\label{app:Jdef}
\end{equation}
Substituting Eq.~\eqref{app:deltav3} gives
\begin{align}
J^\mu_\xi
=
2E^\mu{}_\nu\xi^\nu
-
2\nabla_\nu
\left(
P^{\mu\nu\alpha\beta}\nabla_\alpha\xi_\beta
+
2\xi_\alpha\nabla_\beta
P^{\mu\nu\alpha\beta}
\right).
\label{app:Jdecomposition}
\end{align}
On shell, $E_{\mu\nu}=0$, and the Noether current reduces to a total divergence. The antisymmetric tensor appearing in this divergence therefore defines the gravitational Noether potential,
\begin{equation}
q_\xi^{\mu\nu}
=
P^{\mu\nu\alpha\beta}\nabla_\alpha\xi_\beta
+
2\xi_\alpha\nabla_\beta
P^{\mu\nu\alpha\beta}.
\label{app:q}
\end{equation}
This is precisely the Noether potential quoted in Eq.~\eqref{eq:general-IW-potential} of the
main text.

For Lovelock gravity,
$\nabla_\beta P^{\mu\nu\alpha\beta}=0$, so that
Eq.~\eqref{app:q} immediately reduces to
\begin{equation}
q_\xi^{\mu\nu}
=
P^{\mu\nu\alpha\beta}\nabla_\alpha\xi_\beta ,
\label{app:q-lovelock}
\end{equation}
which is the gravitational Noether potential used in the subsequent analysis.
%%%%%%%%%%%%%%%%%%%%%%%%%%%%%%%%%%%%%%%%%%%%%%%%%%%%%%%%%%

\section{Balance Law in Stationary Axisymmetric Spacetimes}
\label{app:rotation}

In this appendix, we extend the analysis using the gravitational Noether charge to stationary axisymmetric spacetimes. Compared with the static case, rotation introduces an additional contribution. In Einstein gravity, we show that this rotation contribution is nonnegative. In Lovelock gravity, however, the corresponding rotational contribution involves higher-curvature terms, and its sign is not guaranteed.

\subsection{Stationary Axisymmetric Spacetimes}
We consider stationary axisymmetric spacetimes admitting two commuting Killing vector fields, $\xi^\mu$ and $\psi^\mu$. The vector field $\xi^\mu$ generates time translations, while $\psi^\mu$ generates rotational symmetry with closed spacelike orbits. We choose coordinates adapted to these symmetries such that $\xi^\mu=(\partial_t)^\mu$ and $\psi^\mu=(\partial_\phi)^\mu$. We further assume that the spacetime possesses a $t$–$\phi$ reflection symmetry, under which the metric is invariant under the transformation $(t,\phi)\rightarrow(-t,-\phi)$. Therefore, this symmetry allows the metric to be written in the circular form
\begin{equation}
ds^2=-N^2\mathrm{d}t^2+\gamma^2(\mathrm{d}\phi-\omega\,\mathrm{d}t)^2
+q_{AB}\mathrm{d}x^A \mathrm{d}x^B,
 \label{eq:appB_metric}
\end{equation}
where $A,B=1,\ldots,d-1$, and all metric functions are independent of $t$ and $\phi$.

Suppose that there exist two consecutive nondegenerate inner Killing horizons $\mathcal H_1$ and $\mathcal H_2$ inside the event horizon, and denote the spacetime region bounded by them as $\mathcal V$. In this region, we consider the spacelike constant-$t$ hypersurfaces $\Sigma_t$ and perform the ADM decomposition of the metric. The lapse function is $N$, the shift vector is $\beta^a=-\omega\psi^a$.
Let $n^\mu$ be the future-directed unit normal to $\Sigma_t$. The stationary Killing vector $\xi^{\mu}$ can then be decomposed as
\begin{equation}
\xi^\mu=Nn^\mu-\omega\psi^\mu.
\label{eq}
\end{equation}
The extrinsic curvature of $\Sigma_t$, defined by $K_{ab}=\frac12\Lie_n h_{ab}$, is given by
\begin{equation}
\begin{aligned}
K_{ab}
&=\frac{1}{2N}
\left(\partial_t h_{ab}-D_a\beta_b-D_b\beta_a\right)\\
&=\frac{1}{2N}
\left(\psi_aD_b\omega+\psi_bD_a\omega\right),
\end{aligned}
\label{eq:AKab}
\end{equation}
where $D_a$ is the covariant derivative associated with $h_{ab}$. In the second equality, we have used stationarity and axisymmetry, which imply $\partial_t h_{ab}=0$ and $\Lie_\psi h_{ab}=2D_{(a}\psi_{b)}=0$, respectively. Since $\psi^aD_a\omega=0$, the trace of the extrinsic curvature vanishes, $K=0$.

\subsection{Gravitational Noether Charge in Einstein Gravity}
Unlike in static spacetimes, where the Killing horizons are generated by the same Killing vector $\xi^\mu$, the generator of a Killing horizon in a stationary axisymmetric spacetime generally takes the form
\begin{equation}
\chi_i^\mu=\xi^\mu+\Omega_i\psi^\mu,
\label{eq}
\end{equation}
where $\Omega_i$ is the constant angular velocity of the Killing horizon $\mathcal H_i$. Since different horizons may have different angular velocities, their Killing generators need not coincide. This distinction must be taken into account when evaluating the gravitational Noether charge on each horizon. We consider the gravitational Noether charge associated with the Killing vector $\chi_i^\mu$ , evaluated on the bifurcation surface $B_i$. In Einstein gravity, it takes the form
\begin{equation}
\mathcal Q^{\text{EH}}_{\chi_i}[B_i]
=-\frac{1}{16\pi}\int_{B_i}
\nabla^\mu\chi_i^\nu\,\mathrm dS_{\mu\nu}
=\frac{\kappa_i A_i}{8\pi}.
\label{eq:ASG}
\end{equation}
To relate the Noether charges on the two horizons, we introduce an auxiliary vector field
\begin{equation}
k^\mu=\xi^\mu+\omega\psi^\mu=Nn^\mu.
\label{eq}
\end{equation}
On each horizon $\mathcal H_i$, where $\omega=\Omega_i$, this vector field $k^{\mu}$ reduces to the corresponding Killing generator $\chi_i^\mu$. The difference between the gravitational Noether charges on the two bifurcation surfaces can therefore be expressed as
\begin{equation}
\begin{aligned}
&\mathcal Q^{\text{EH}}_{\chi_1}[B_1]-\mathcal Q^{\text{EH}}_{\chi_2}[B_2]\\
&=-\frac{1}{16\pi}
\int_{B_1-B_{2}}\left(\nabla^\mu\xi^\nu+\omega\nabla^\mu\psi^\nu\right)\mathrm dS_{\mu\nu}.
\end{aligned}
\label{eq}
\end{equation}
Applying Stokes' theorem to the hypersurface $\mathcal D=\Sigma_t\cap\mathcal V$, we obtain
\begin{equation}
\begin{aligned}
&\mathcal Q^{\text{EH}}_{\chi_1}[B_1]-\mathcal Q^{\text{EH}}_{\chi_2}[B_2]\\
&=-\frac{1}{8\pi}\int_{\mathcal D}
\nabla_\nu\left(
\nabla^\mu\xi^\nu+\omega\nabla^\mu\psi^\nu
\right)\mathrm d\Sigma_\mu\\
&=\frac{1}{8\pi}\int_{\mathcal D}
\left[
R_{\mu\nu}n^\mu k^\nu
+n_\mu(\nabla_\nu\omega)\nabla^\mu\psi^\nu
\right]\mathrm d\Sigma,
\end{aligned}
\label{eq:appB_stokes}
\end{equation}
where $d\Sigma_\mu=-n_\mu d\Sigma$, and we have used the Killing equation for $\xi^\mu$ and $\psi^\mu$ here. We now evaluate the additional rotational contribution. Using the Killing equation for $\psi^\mu$ and the orthogonality relation $n_\mu\psi^\mu=0$, we find
\begin{equation}
\begin{aligned}
n_\mu(\nabla^\mu\psi^\nu)(\nabla_\nu\omega)
&=\psi^\mu(\nabla^\nu n_\mu)(\nabla_\nu\omega)\\
&=\psi^a e^\mu{}_a e^\nu{}_b
(\nabla_\nu n^\mu)D^b\omega\\
&=\psi^a K_{ab}D^b\omega\\
&=\frac{\psi^2}{2N}D_a\omega D^a\omega\\
&=\frac{\gamma^2}{2N}
q^{AB}(\partial_A\omega)(\partial_B\omega)
\geq0.
\end{aligned}
\label{eq}
\end{equation}
Here, $e^\mu{}a$ denotes the tangent basis vectors of $\Sigma_t$ in spacetime, and we have used the definition $K_{ab}=e^\mu{}_a e^\nu{}_b\nabla_\nu n_\mu$, together with its explicit expression given in Eq.~\eqref{eq:AKab}. Combining Eqs. (B5), (B8), and (B9), we obtain
\begin{equation}
\begin{aligned}
&\frac{\kappa_1 A_1}{8\pi}
-\frac{\kappa_2 A_2}{8\pi}\\
&=\frac{1}{8\pi}\int_{\mathcal D}
\left[
NR_{\mu\nu}n^\mu n^\nu
+\frac{\gamma^2}{2N}q^{AB}
(\partial_A\omega)(\partial_B\omega)
\right]\mathrm d\Sigma\\
\end{aligned}
\label{eq:Abalancelaw}
\end{equation}
It is worth noting that the bulk integrand is precisely the spatial Laplacian of the lapse function $D^2N$. Thus, this balance law provides an integral formulation of the lapse equation, even in the presence of rotation.

In Einstein gravity, the rotational correction of stationary axisymmetric spacetimes in Eq.~\ref{eq:Abalancelaw} is manifestly nonnegative. Under the TCC, which follows from the SEC through Einstein's equations, the right hand side of Eq.~\ref{eq:Abalancelaw} is therefore nonnegative. On the other hand, the surface gravities of two consecutive nondegenerate inner Killing horizons satisfy $\kappa_1<0<\kappa_2$, making the left hand side strictly negative. This contradiction excludes the existence of two or more nondegenerate inner Killing horizons, showing that the gravitational Noether charge method can also be used to prove the inner horizon constraint in stationary axisymmetric spacetimes of Einstein gravity.

\subsection{Gravitational Noether Charge in Lovelock Gravity}

We now extend the above analysis to Lovelock gravity. As shown in the main text, the gravitational Noether charge associated with the Killing generator $\chi_i^\mu$ on the bifurcation surface $B_i$ is
\begin{equation}
\begin{aligned}
\mathcal Q_{\chi_i}[B_i]
&=-\frac{1}{16\pi}\int_{B_i}
P^{\mu\nu\rho\sigma}\nabla_\rho\chi_{i\sigma}\,\mathrm dS_{\mu\nu}\\
&=\frac{\kappa_i}{2\pi}S_{\rm W}[B_i].
\end{aligned}
\label{eq}
\end{equation}
Here, $P^{\mu\nu\rho\sigma}=\partial L_{\rm grav}/\partial R_{\mu\nu\rho\sigma}$ is the curvature response tensor of Lovelock gravity.

Following the same procedure as in Einstein gravity, we consider the difference between the Noether charges on the two bifurcation surfaces. Applying Stokes' theorem to $\mathcal D$, we obtain
\begin{equation}
\begin{aligned}
&\mathcal Q_{\chi_1}[B_1]
-\mathcal Q_{\chi_2}[B_2]\\
&=-\frac{1}{8\pi}\int_{\mathcal D}
\nabla_\nu\Big[
P^{\mu\nu\rho\sigma}
\big(\nabla_\rho\xi_\sigma
+\omega\nabla_\rho\psi_\sigma\big)
\Big]d\Sigma_\mu\\
&=\frac{1}{8\pi}\int_{\mathcal D}
\left(
\mathcal R_{\mu\nu}n^\mu k^\nu
+n_\mu(\nabla_\nu\omega)
P^{\mu\nu\rho\sigma}\nabla_\rho\psi_\sigma
\right)d\Sigma,
\end{aligned}
\label{eq}
\end{equation}
where we have used the Lovelock identity $\nabla_\mu P^{\mu\nu\rho\sigma}=0$ and the Killing equations for $\xi^\mu$ and $\psi^\mu$. The additional rotational contribution is
\begin{equation}
\mathcal I_{\rm rot}
=n_\mu(\nabla_\nu\omega)
P^{\mu\nu\rho\sigma}\nabla_\rho\psi_\sigma.
\label{eq:Arot}
\end{equation}
%%%%%%%%%%%%%%%%%%%%%%%%%%%%
We next examine this rotational contribution. We perform the $(d+1)$ decomposition of $\nabla_\rho\psi_\sigma$ with respect to $\Sigma_t$. Using the Killing equation for $\psi^\mu$ and $n_\mu\psi^\mu=0$, we obtain
\begin{equation}
\begin{aligned}
    \nabla_\rho\psi_\sigma&=(h_\rho{}^\alpha-n_\rho n^\alpha)
(h_\sigma{}^\beta-n_\sigma n^\beta) \nabla_\alpha\psi_\beta\\
&=D_\rho\psi_\sigma
-n_\rho K_{\sigma a}\psi^a
+n_\sigma K_{\rho a}\psi^a.
\end{aligned}
\label{eq}
\end{equation}
Substituting this decomposition into Eq.~\eqref{eq:Arot}, and using the explicit expression for $K_{ab}$ in Eq.~\eqref{eq:AKab} together with $\psi^\mu\nabla_\mu\omega=0$ and $\psi^\mu\psi_\mu=\gamma^2$, we obtain
\begin{equation}
\begin{aligned}
\mathcal I_{\rm rot}
={}&n_\mu P^{\mu\nu\rho\sigma}
(\nabla_\nu\omega)D_\rho\psi_\sigma\\
&-\frac{\gamma^2}{N}
n_\mu n_\rho P^{\mu\nu\rho\sigma}
(\nabla_\nu\omega)(\nabla_\sigma\omega).
\end{aligned}
\label{eq:appB_rot_reduction}
\end{equation}

In Einstein gravity, the curvature response tensor takes the form
\begin{equation}
P^{\mu\nu\rho\sigma}
=\frac12\left(g^{\mu\rho}g^{\nu\sigma}
-g^{\mu\sigma}g^{\nu\rho}\right).
\end{equation}
The first term in Eq.~\eqref{eq:appB_rot_reduction} vanishes, while $n_\mu n_\rho P^{\mu\nu\rho\sigma}=-h^{\nu\sigma}/2$. Consequently, the rotational contribution reduces to
\begin{equation}
\mathcal I_{\rm rot}
=\frac{\gamma^2}{2N}D_\mu\omega D^\mu\omega
\geq0,
\end{equation}
recovering the result obtained above. In Lovelock gravity, however, $P^{\mu\nu\rho\sigma}$ contains higher-curvature contributions. The first term in Eq.~\eqref{eq:appB_rot_reduction} need not vanish, and the second term is no longer necessarily nonnegative. Therefore, the sign of $\mathcal I_{\rm rot}$ cannot be determined from the TCC-like curvature condition and the positivity of Wald entropy alone.

\bibliography{references}

\end{document}